\documentclass[letterpaper, 10pt, conference]{ieeeconf}      

\IEEEoverridecommandlockouts
\usepackage[
style=ieee,    
backend=bibtex,   
isbn=false,       
doi=false,        
giveninits=true,  
url=false,
date=year,        
maxbibnames=99,   
minnames=1,       
maxnames=2,      
sorting=none,     
alldates=year,
]{biblatex}           
\bibliography{bibliography_ecc.bib}       

\AtEveryBibitem{\clearfield{pages}} 
\AtEveryBibitem{\clearfield{volume}}
\AtEveryBibitem{\clearfield{number}}  

\usepackage{csquotes} 

\usepackage{glossaries}
\usepackage{booktabs}

\usepackage{siunitx}				
\usepackage{amsmath,amssymb,amsfonts}

\usepackage{color}								
\usepackage{xcolor}
\usepackage{graphicx} 
\usepackage{tikz}
\usetikzlibrary{calc}
\usetikzlibrary{arrows}

\usepackage[german,USenglish,UKenglish]{babel}

\usepackage{enumerate}				
\usepackage[hidelinks]{hyperref}    

\usepackage[ruled, vlined, linesnumbered,english]{algorithm2e}
\SetKwRepeat{Do}{do}{while}%
\SetAlFnt{\small} 
\SetKwComment{Comment}{$\triangleright$\ }{}
\SetVlineSkip{2pt} 
\DontPrintSemicolon
\SetNlSty{}{}{:} 
\SetAlgoNlRelativeSize{-1} 
\SetCommentSty{} 

\newcommand{\col}[1]{\text{col}\{#1\}}
\newcommand{\diag}[1]{\text{diag}\{#1\}}

\newcommand{\mat}[1]{\begin{bmatrix}#1\end{bmatrix}}
\newcommand{\R}{\mathbb{R}}

\newcommand{\spec}[1]{\text{Spec}(#1)}

\newcommand{\xat}[2]{\ifthenelse{\equal{#2}{}}{{x}_{#1}}{{x}_{#1}(#2)}}
\newcommand{\uat}[2]{\ifthenelse{\equal{#2}{}}{{u}_{#1}}{{u}_{#1}(#2)}}

\newcommand{\vi}{v_i}
\newcommand{\ifi}{i_{\text{f},i}}
\newcommand{\ei}{e_i}
\newcommand{\dvi}{\dot{v}_i}
\newcommand{\difi}{\dot{i}_{\text{f},i}}
\newcommand{\dei}{\dot{e}_i}
\newcommand{\vti}{v_{\text{t},i}}
\newcommand{\iexti}{i_{{\text{ext}},i}}
\newcommand{\Cfi}{C_{\text{f},i}}
\newcommand{\Rfi}{R_{\text{f},i}}
\newcommand{\Lfi}{L_{\text{f},i}}
\newcommand{\Yad}{Y_{\text{ad}}}

\newcommand{\Cf}{C_{\text{t}}}
\newcommand{\Rf}{R_{\text{t}}}
\newcommand{\Lf}{L_{\text{t}}}
\newcommand{\B}{M}

\newcommand{\iLall}{i_{\text{L}}}
\newcommand{\ifall}{i_{\text{f}}}
\newcommand{\vall}{v}
\newcommand{\eall}{e}
\newcommand{\diLall}{\dot{i}_{\text{L}}}
\newcommand{\difall}{\dot{i}_{\text{f}}}
\newcommand{\dv}{v_{\Delta}}
\newcommand{\dvall}{\dot{v}}
\newcommand{\deall}{\dot{e}}
\newcommand{\iextall}{i_{{\text{ext}}}}

\newcommand{\Rlj}{R_{\text{l},j}}
\newcommand{\Llj}{L_{\text{l},j}}
\newcommand{\Rl}{R_{\text{l}}}
\newcommand{\Ll}{L_{\text{l}}}
\newcommand{\dvj}{v_{\Delta j}}
\newcommand{\iLi}{i_{\text{L},j}}
\newcommand{\diLi}{\dot{i}_{\text{L},i}}

\newcommand{\viRL}{\tilde{v}_i}
\newcommand{\iextiRL}{\tilde{i}_{{\text{ext}},i}}
\newcommand{\dvjRL}{\tilde{v}_{\Delta j}}
\newcommand{\iLiRL}{\tilde{i}_{\text{L},j}}
\newcommand{\vallRL}{\tilde{v}}
\newcommand{\iextallRL}{\tilde{i}_{{\text{ext}}}}
\newcommand{\dvallRL}{\tilde{v}_{\Delta }}
\newcommand{\iLallRL}{\tilde{i}_{\text{L}}}

\newcommand{\vref}{{v}_{\text{ref}}}
\newcommand{\vopt}{{v}^{*}}

\newcommand{\numOfnodes}{n}
\newcommand{\numOflines}{m}

\newcommand{\kinit}{1}
\newcommand{\kend}{N}
\newcommand{\Qloss}{{Q}_{\text{loss}}}
\newcommand{\ssopt}{\text{Opt}_{\text{ss}}}
\newcommand{\mpctrack}{\text{MPC}_{\text{track}}}
\newcommand{\empc}{\text{MPC}_{\text{econ}}}
\newcommand{\Pconst}{\mathbb{P}}
\newcommand{\Vconst}{\mathbb{V}}
\newcommand{\Xconst}{\mathbb{X}}
\newcommand{\Uconst}{\mathbb{U}}

\newtheorem{theorem}{Theorem}

\newtheorem{remark}{Remark}
\newtheorem{proposition}{Proposition}

\newacronym{mpc}{MPC}{Model Predictive Control}
\newacronym{res}{RES}{renewable energy sources}
\newacronym{dgu}{DGU}{distributed generation unit} 
\newacronym{mg}{MG}{microgrid} 

\title{\LARGE \bf
On MPC-based Strategies for Optimal Voltage References in DC Microgrids
}

\author{Pol Jané-Soneira, Ionela Prodan, Albertus J. Malan, and Sören Hohmann
	\thanks{P. Jané-Soneira, Albertus J. Malan and S. Hohmann are with the Institute of Control Systems, Karlsruhe Institute of Technology (KIT), 76131, Karlsruhe, Germany. Corresponding author is P. Jané-Soneira, {\tt \scriptsize pol.jane@kit.edu}. This work has been partly funded by Germany’s Federal Ministry for Economic Affairs and Energy within the project RegEnZell (reference number 0350062C)}%
	\thanks{I. Prodan is with the Univ. Grenoble Alpes, Grenoble INP$^\dagger{}$, LCIS, F-26000, Valence, France, (e-mail: ionela.prodan@lcis.grenoble-inp.fr). $^\dagger$Institute of Engineering and Management Univ. Grenoble Alpes. I. Prodan's research benefited from the support of the FMJH Program PGMO and from the support to this program from EDF.}%
}

\begin{document}

\maketitle
\thispagestyle{empty}
\pagestyle{empty}

\begin{abstract}
    Modern power systems are characterized by low inertia and fast voltage dynamics due to the increase of sources connecting via power electronics and the removal of large traditional thermal generators. Power electronics are commonly equipped with fast controllers that are able to reach a desired voltage setpoint within seconds. In this paper, we propose and compare two approaches using Model Predictive Control (MPC) to compute optimal voltage references for the power electronic devices in order to minimize the losses in a DC microgrid: i) a traditional setpoint-tracking MPC which receives a previously computed optimal setpoint; ii) an economic MPC which does not require a priori computed setpoints. We show that the economic MPC outperforms the setpoint-tracking MPC in simulations with the CIGRE benchmark system when multiple load disturbances occur. Some insights and discussions related to the stability of the closed-loop system using its dissipativity properties are highlighted for both approaches.
\end{abstract}

\section{Introduction} \label{sec:intro}

In the last years, DC \glspl{mg} have become technically feasible due to the recent advances in semiconductor converter technology \cite{elsayed2015dc}. Furthermore, they have the potential to prevail over their AC counterparts in the future due to their higher efficiency, the more natural interface to most \glspl{dgu}, \glspl{res}, storages and loads \cite{elsayed2015dc, justo2013ac}. 
In addition, DC \glspl{mg} are significantly simpler to regulate, since the frequency control becomes unnecessary \cite{justo2013ac}. The problem of maintaining constant voltage levels in DC \glspl{mg} under varying load conditions is well studied in literature and it is called primary control. On the one hand, droop-based methods \cite{guerrero2010hierarchical,zhao2015distributed} are widely-used decentralized approaches and exhibit favorable properties such as (limited) power-sharing. Several improvements such as nonlinear, adaptive or dead-band droop have also been proposed, as summarized in \cite{gao2019}. However, these methods show load-dependent voltage deviation and steady-state voltage offsets, which need to be compensated by a higher level control.
On the other hand, passivity-based controllers tackling the shortcomings of droop-based approaches have been proposed recently \cite{nahata2020passivity,laib2023decentralized}. These regulators achieve an offset-free regulation of a given voltage reference and exhibit advantageous plug-and-play properties for \glspl{dgu} while guaranteeing overall asymptotic stability via passivity. Similar stability properties can only be achieved with droop-based methods by simplifying the system models with questionable assumptions and approximations \cite{zhao2015distributed}. However, these passivity-based controllers necessitate a secondary control to achieve power-sharing or coordination.  

Secondary control architectures, built on top of droop-controlled or passivity-based controlled \glspl{dgu}, differ considerably in their purpose and nature. The secondary control of droop-controlled \glspl{dgu} is typically a proportional-integral controller, which is used for compensating unwanted voltage drifts. It can be implemented in a centralized, decentralized or distributed manner using consensus protocols \cite{gao2019,simpson2015dapi}. In contrast to droop-based methods, passivity-based primary controllers do not necessitate a secondary control layer in order to restore the desired or feasible voltage levels \cite{nahata2020passivity,laib2023decentralized,strehle2020}. With a secondary control layer freed from ensuring stability and feasibility, it can focus on providing suitable voltage references in order to pursue other objectives, such as minimizing power losses or costs, or achieving power-sharing. In \cite{nahata2022current}, a secondary controller for achieving proportional current sharing is proposed. Similarly, \cite{malan2022distributed} proposes DC power sharing whilst also taking unactuated agents into account. These methods are able to guarantee convergence to a desirable steady-state employing passivity theory. However, as it is the case for the secondary control architectures for droop-based controllers, the dynamic performance for reaching this steady-state is not considered, i.e., voltage overshoots damaging \gls{mg} components may occur. Furthermore, similar to the primary passivity-based controllers, they cannot take into account actuator or state constraints. 

\gls{mpc} for DC \glspl{mg} is an emerging research topic since it allows considering state and input constraints \cite{hu2021model}. Many approaches consider an \gls{mpc}-based secondary control for droop-controlled DC \glspl{mg} \cite{karami2021,vu2015model,kishore2018efficient,guannan2017}. However, since they require droop-based primary controllers, these approaches do not allow a plug-and-play operation of \glspl{dgu}. Other approaches attempt to directly control the Buck converter voltage without an underlying primary controller \cite{noroozi2018model}. However, the filter dynamics necessitate a sampling time in the microseconds range, which makes an online solution of the \gls{mpc} optimization problem feasible only for small \glspl{mg} and short optimization horizons \cite{noroozi2018model}. Recently, economic MPC has been applied to power systems in \cite{jia2020optimal,kohler2017real}. However, these approaches consider simplified AC systems, and the results are not applicable to low inertia DC power systems as considered in this work. 


\emph{Contributions:} We propose a secondary control architecture using \gls{mpc} for DC \glspl{mg} with passivity-based primary controllers, which is able to guarantee constraint satisfaction while minimizing an objective function. In particular, we propose two different receding horizon controller designs: i) a classic setpoint-tracking \gls{mpc}, for which we propose a Lyapunov function for any \gls{mg} configuration and hence proof stability using established methods; ii) an economic \gls{mpc} with better closed-loop performance. Furthermore, we provide extensive simulation results on a realistic benchmark system and a thorough comparison of both MPC-based approaches.

The remainder of this paper is structured as follows. In Section~2, the system model for which we design the \gls{mpc}-based secondary control, comprising the physical system endowed with passivity-based primary controllers, is presented. The tracking and economic \gls{mpc} designs are described in Section~3. In Section~4, we present extensive simulation results comparing both approaches in the nominal case and when disturbances are present. Finally, we draw conclusions and outline future research directions in Section~5.  

\emph{Notation:} %
Lowercase letters $x \in \R^{n}$ represent vectors, and uppercase letters $X \in \R^{n\times n}$ represent matrices. The transpose of a vector $x \in \R^n$ is written as $x^\top$. The vector $x = \col{x_i}$ and matrix $X = \diag{x_i}$ are the $n\times 1$ column vector and $n \times n$ diagonal matrix of the elements $x_i$, $i = 1,\dots, n$, respectively. Let $I_{n}$ denote the $n \times n$ identity matrix and $\spec{X}$ the spectrum of matrix $X$. A directed graph is denoted by $\mathcal{G}(\mathcal{V},\mathcal{E})$, where $\mathcal{V}$ is the set of nodes and $\mathcal{E} \subset \mathcal{V}\times \mathcal{V}$ the set of edges. The cardinality for a set $\mathcal{V}$ is denoted by $\vert \mathcal{V} \vert$. The incidence matrix $M \in \R^{\vert \mathcal{V} \vert \times \vert \mathcal{E} \vert }$ is defined as $M = (m_{ij})$ with $m_{ij} = -1$ if edge $e_i \in \mathcal{E}$ leaves node $v_j \in \mathcal{V}$, $m_{ij} = 1$ if edge $e_i \in \mathcal{E}$ enters node $v_j \in \mathcal{V}$, and $m_{ij} = 0$ otherwise. The Hadamard (element-wise) product of $x$ and $y$ is denoted by $x \circ y$.


\section{System Model}

In this paper, we consider a set $\mathcal{V}$ of $\numOfnodes = \vert \mathcal{V} \vert$ electrical buses connected via a set $\mathcal{E}$ of $\numOflines = \vert \mathcal{E} \vert$ electrical lines. The network topology is described with the graph $\mathcal{G}(\mathcal{V},\mathcal{E})$, where $\mathcal{V}$ are the nodes and $\mathcal{E}$ the edges. Fig.~\ref{fig:cigre} shows such a \gls{mg} with $n=11$ nodes and $m=12$ lines. In the next subsections, we derive the models of the components of the DC \gls{mg}, i.e. the electrical buses and the power lines. In Subsection~\ref{subsec:overall}, we show favorable passivity properties of the overall \gls{mg} system. To conclude the section, we propose a model reduction, discretize the system and provide a Lyapunov function which holds for any \gls{mg} topology, which is necessary for the MPC design in Section~\ref{sec:mpc}.
\begin{figure}[htbp]
    \centering
    \includegraphics[scale=0.9]{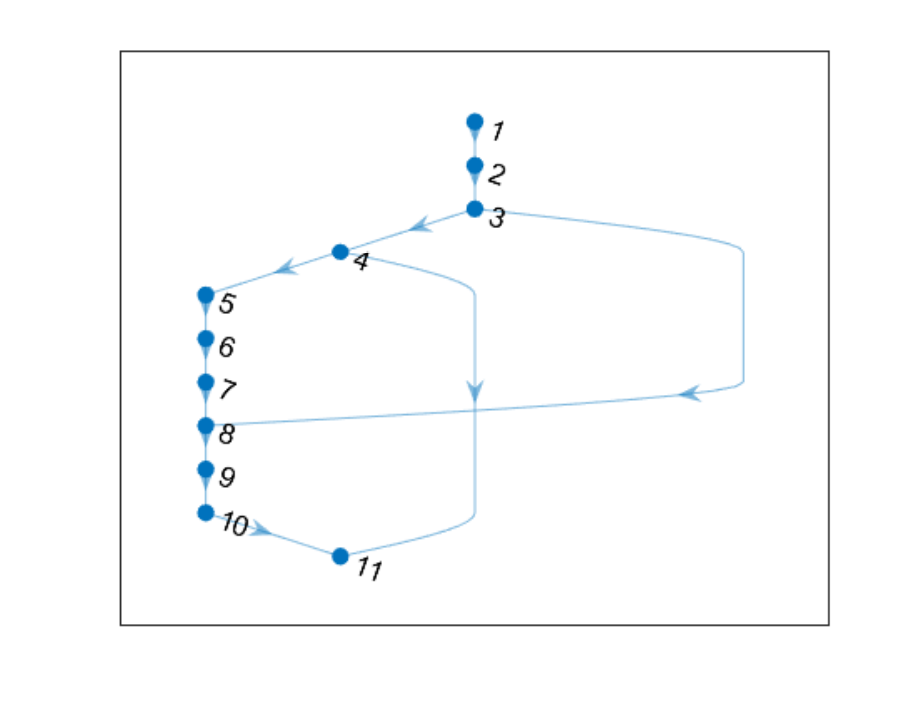}
    \vspace*{-0.5cm}
    \caption{Graphical representation of the meshed CIGRE benchmark \gls{mg} \cite{rudion2006design} with $n = 11$ nodes and $m = 12$ power lines.}
    \label{fig:cigre}
\end{figure}

\subsection{Distributed generation unit (DGU)}
\begin{figure}[ht!]
    \centering
    \includegraphics{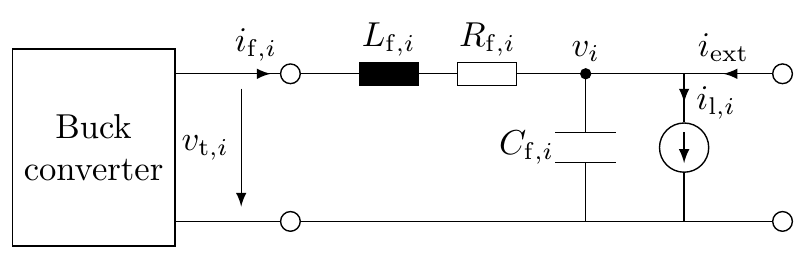}
    \caption{Electric scheme of a bus including a \gls{dgu} and a load.}
    \label{fig:dgu}
\end{figure}
A graphical representation of a bus is shown in Fig.~\ref{fig:dgu}. It is composed of a Buck converter supplying a voltage $\vti \in \R_{\geq 0}$, a load with current $i_{{\rm l},i} \in \R$, and a filter with the resistance $\Rfi \in \R_{> 0}$, inductance $\Lfi \in \R_{>0}$ and capacitance $\Cfi \in \R_{> 0}$. We idealize the Buck converter with the widespread averaging model, by which we disregard the switching behavior~\cite{middlebrook1977}. The dynamics for every bus $i \in \mathcal{V}$ are described with the states $\vi$ and $\ifi$, which describe the node voltage and filter current, by
\begin{subequations}
    \begin{align}
        \Cfi \dvi & = \ifi - i_{{\rm l},i}(\vi) - \iexti \\   
        \Lfi \difi & = \vti -\Rfi \ifi - \vi ,
    \end{align}
\end{subequations}
where $\iexti$ is the cumulative current injected by interconnecting lines. The \glspl{dgu} at the nodes are equipped with a passivity-based primary voltage controller as in \cite{nahata2020passivity}, which regulates $\vi$ to a reference voltage $v_{\text{ref},i}$. The controller adds a state $\ei \in \R$ to the system for the integral action, and employs a state feedback $\vti = k_i^\top [\vi, \ifi, \ei]^\top$ with $k_i \in \R^{3}$. The node dynamics with the passivity-based voltage controller are
\begin{subequations}\label{eq:bus}
    \begin{align}
        \Cfi \dvi & = \ifi - i_{{\rm l},i}(\vi) - \iexti \\   
        \Lfi \difi & = \alpha_i \vi + \beta_i \ifi + \gamma_i \ei \\
        \dei & =  v_{\text{ref},i} - \vi, \label{eq:gridFeeding_nonlinearity}
    \end{align}
\end{subequations}
where 
\begin{align*}
    \alpha_i = \frac{k_{1,i}-1}{\Lfi}, \; \beta_i = \frac{k_{2,i}-\Rfi}{\Lfi}, \; \gamma_i = \frac{k_{3,i}}{\Lfi}.
\end{align*}
are controller parameters. In this work, we assume a time-varying resistive load, i.e. 
\begin{equation} \label{eq:load}
    i_{{\rm l},i}(\vi) = Y_i(t)\vi,
\end{equation}
where $Y_i \in \R>0$ is the load admittance. In the following, we present an important result from \cite{nahata2020passivity} which will prove instrumental for the \gls{mpc} stability analysis. It states that choosing the control parameters according to \cite{nahata2020passivity}, the node dynamics are passive w.r.t. the input-output pair ($-\iexti$,$\vi$). This allows to interconnect arbitrarily many nodes by (passive) lines while guaranteeing stability for constant voltage references $v_{\text{ref},i}$. 
%
\begin{proposition}[\cite{nahata2020passivity}] \label{prop:bus}
    Let $[\bar{v}_i,\bar{i}_{{\rm f},i},\bar{e}_i]$ be an equilibrium point of~\eqref{eq:bus} and define the error variables $\tilde{v}_i = {v}_i - \bar{v}_i$, $\tilde{i}_{{\rm f},i} = {i}_{{\rm f},i} - \bar{i}_{{\rm f},i}$ and $\tilde{e}_i = e_i - \bar{e}_i$. System~\eqref{eq:bus} is equilibrium independent passive w.r.t. the input-output pair $(-\iextiRL,\viRL)$  
    for any $v_{\text{ref},i} > 0$ with the storage function $S_i: \R^3 \rightarrow \R_{\geq 0}$
    \begin{align}
        S_i(\tilde{v}_i,\tilde{i}_{{\rm f},i},\tilde{e}_i) = \mat{\tilde{v}_i \\ \tilde{i}_{{\rm f},i} \\ \tilde{e}_i}^\top \mat{\Cfi & 0 & 0 \\ 0 & \frac{\beta_i}{\omega_i} & \frac{\gamma_i}{\omega_i} \\ 0 & \frac{\gamma_i}{\omega_i} & \frac{\alpha_i \gamma_i}{\omega_i} } \mat{\tilde{v}_i \\ \tilde{i}_{{\rm f},i} \\ \tilde{e}_i} 
    \end{align}
    if the control parameters are chosen such that
    \begin{align}
        k_{1,i} &< 1 \\
        k_{2,i} &< \Rfi \\
        0 < k_{3,i} &< \frac{1}{\Lfi} (k_{1,i} - 1)(k_{2,i} - \Rfi).
    \end{align}
\end{proposition}
\begin{proof}
    The proof can be found in~\cite{nahata2020passivity}.
\end{proof}

\subsection{Power Line}

The power lines are modeled with the pi equivalent circuit~\cite{machowski2020power}, which is shown in Fig.~\ref{fig:line}. It is composed of a series inductance $\Llj \in \R_{>0}$ and resistance $\Rlj\in \R_{>0}$, and two parallel capacitances $\frac{C_{{\rm l},j}}{2}\in \R_{>0}$. Note that the line capacitance is connected in parallel to the bus filter capacitance of the busses which the line is interconnecting. Hence, only one capacitor with a capacitance being the sum of both can be considered. Furthermore, since the capacitance of typical filter capacitors \cite{tucci2016,nahata2020passivity} are higher than line capacitances of medium voltage power lines \cite{rudion2006design} by several orders of magnitude, the line capacitors can be neglected. The dynamics for line $j \in \mathcal{E}$ interconnecting nodes $k,l \in \mathcal{V}$ are thus described by 
\begin{align}\label{eq:line_dynamics} 
    \Llj \diLi = -\Rlj \iLi + \dvj,
\end{align}
where $\dvj = v_l - v_k$ is the input and $v_k$, $v_l$ are the bus voltages. The next proposition serves as an important building block towards an overall Lyapunov function in Subsection~\ref{subsec:overall}. 
\begin{figure}[ht!]
    \centering
    \includegraphics{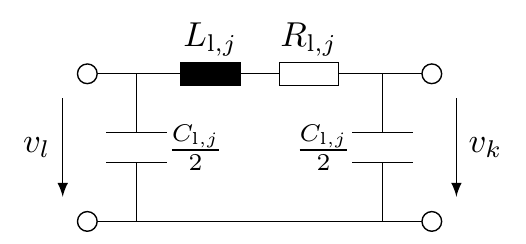}
    \caption{Graphical representation of a line in the pi equivalent circuit.}
    \label{fig:line}
\end{figure}

\begin{proposition}[\cite{brogliato2007dissipative}] \label{prop:lines}
    Let $\bar{i}_{{\rm L},j}, \bar{v}_{\Delta,j}$ be an equilibrium point of~\eqref{eq:line_dynamics} and define $\tilde{i}_{{\rm L},j} = \iLi - \bar{i}_{{\rm L},j}$, $ \tilde{v}_{\Delta,j} = \dvj- \bar{v}_{\Delta,j}$. System~\eqref{eq:line_dynamics} is equilibrium independent strictly passive w.r.t. the input-output pair $(\iLiRL,\dvjRL)$  
    with the storage function $S_j: \R \rightarrow \R_{\geq 0}$
    \begin{align}
        S_j(\tilde{i}_{{\rm L},j}) = \Llj \tilde{i}_{{\rm L},j}^2.
    \end{align}
\end{proposition}
\begin{proof}
    The proof is given in \cite{brogliato2007dissipative}.
\end{proof}
\begin{remark}\label{rem:static_lines}
    For static power lines, i.e. if the inductance is neglected, system~\eqref{eq:line_dynamics} is still strictly passive w.r.t. the same input-output pair, since it represents an offset-free strictly monotonically increasing map~\cite{khalil2002nonlinear}. 
\end{remark}

In the next subsection, the stability of the overall system is investigated, after which we dwell on the model reduction and discretization for using the model for the MPC design. 

\subsection{Overall System} \label{subsec:overall}

A group of $\numOfnodes = \vert \mathcal{V} \vert$ independent buses is described by
\begin{subequations} \label{eq:independent_bus}
    \begin{align}
        \Cf \dvall & = \ifall - Y(t) \vall  - \iextall \\   
            \Lf \difall & = \alpha \vall + \beta \ifall + \gamma \eall \\
            \deall & = v_{\text{ref}} - \vall
    \end{align}
\end{subequations}
where $\alpha = \diag{\alpha_i}$, $\beta = \diag{\beta_i}$ and $\gamma = \diag{\gamma_i}$ contain the control parameters, $\Cf = \diag{\Cfi}$, $\Rf= \diag{\Rfi}$, $\Lf= \diag{\Lfi}$ and $Y = \diag{Y_i}$ are the filter, load and line parameters, and $\vall = \col{\vi}$, $\ifall = \col{\ifi}$ and $\eall = \col{\ei}$ are the stacked states of each DGU $i \in \mathcal{V}$. A group of $\numOflines = \vert \mathcal{E} \vert$ power lines is described by the equations
\begin{align} \label{eq:independent_lines}
    \Ll \diLall &= -\Rl \iLall + \dv,
\end{align}
where $\iLall = \col{\iLi}$, $\dv = \col{\dvj}$ are the stacked states and inputs, and $\Rl = \diag{\Rlj}$ and $\Ll = \diag{\Llj}$ the parameters of each power line $j \in \mathcal{E}$. The interconnection of the buses $i \in \mathcal{V}$ and the power lines $j \in \mathcal{E}$ according to the graph $\mathcal{G}(\mathcal{V},\mathcal{E})$ is described by the incidence matrix $M \in \R^{\numOfnodes \times \numOflines}$ as defined in Section~\ref{sec:intro}. In particular, the voltage drop over the power lines $\dv \in \R^{\numOflines}$ can be described with the voltages of the buses they are connected to, i.e. 
\begin{equation}\label{eq:interconnection1}
    \col{\dvj} = \dv = M^\top \vall = M^\top \col{\vi}.
\end{equation}
The current drawn from bus $i \in \mathcal{V}$ is the sum of the currents through lines connected to bus $i$, i.e.,
\begin{equation} \label{eq:interconnection2}
    \col{-\iexti} = - \iextall =  -M \iLall = -M \col{\iLi}.
\end{equation}
According to \cite[Lemma~1]{nahata2020passivity}, this constitutes a skew-symmetric interconnection. Having described the interconnection, the whole \gls{mg} composed of $\numOfnodes = \vert \mathcal{V} \vert$ nodes interconnected by $\numOflines = \vert \mathcal{E} \vert$ lines obeys the dynamics
\begin{subequations} \label{eq:microgrid_dynamics}
    \begin{align}
        \Cf \dvall & = \ifall - Y(t) \vall  - \B \iLall \label{eq:eq1} \\   
        \Lf \difall & = \alpha \vall + \beta \ifall + \gamma \eall \\
        \deall & = v_{\text{ref}} - \vall \\
        \Ll \diLall &= -\Rl \iLall + \B^T \vall. \label{eq:line_dynamics_1}
    \end{align}
\end{subequations}
The following result about the stability of the interconnected system paves the way for proving stability of the closed-loop system with \gls{mpc}.

\begin{proposition}\label{th:lyapunov_timeCont}
    Consider a system of independent buses and lines as in \eqref{eq:independent_bus}~and~\eqref{eq:independent_lines} interconnected through \eqref{eq:interconnection1}~and~\eqref{eq:interconnection2} as in~\eqref{eq:microgrid_dynamics}. A Lyapunov function of the interconnected system for any $\vref>0$ is given by
    \begin{equation} \label{eq:lyapunov_func}
        V(\tilde{x}) = \sum_{i \in \mathcal{V}} S_i(\tilde{v}_i,\tilde{i}_{{\rm f},i},\tilde{e}_i) + \sum_{j \in \mathcal{E}} S_j(\tilde{i}_{{\rm L},j}) 
    \end{equation}
    with $x = \col{\tilde{v}_i,\tilde{i}_{{\rm f},i},\tilde{e}_i,\tilde{i}_{{\rm L},j}}$. 
\end{proposition}
\begin{proof}
    System~\eqref{eq:independent_bus} is passive w.r.t. $(-\iextallRL, \vallRL)$ with storage function $S_{\rm bus} = \sum S_i(\tilde{v}_i,\tilde{i}_{{\rm f},i},\tilde{e}_i)$, which follows trivially from Prop.~\ref{prop:bus}, since the bus dynamics in~\eqref{eq:independent_bus} are independent. Applying the same reasoning, system~\eqref{eq:independent_lines} is strictly passive w.r.t. $(\dvallRL, \iLallRL)$ with storage function $S_{\rm lines} = \sum S_j(\tilde{i}_{{\rm L},j})$ by Prop.~\ref{prop:lines}. From \eqref{eq:lyapunov_func}, it holds that
    \begin{subequations}
        \begin{align}
            \dot{V} & = \sum_{i \in \mathcal{V}} \dot{S}_i(\tilde{v}_i,\tilde{i}_{{\rm f},i},\tilde{e}_i) + \sum_{j \in \mathcal{E}} \dot{S}_j(\tilde{i}_{{\rm L},j}) \\
            & < -\iextallRL^\top \vallRL + \iLiRL^\top \dvjRL \stackrel{\eqref{eq:interconnection1},\eqref{eq:interconnection2}}{=} 0. \label{eq:lyapunov_ineq}
        \end{align}
    \end{subequations}
    The inequality in~\eqref{eq:lyapunov_ineq} follows from the passivity of~\eqref{eq:independent_bus} and the strict passivity of~\eqref{eq:independent_lines} as in Prop.~\ref*{prop:bus}~and~Prop.~\ref*{prop:lines}. The last equality follows from the skew symmetric interconnection~\eqref{eq:interconnection1}~and~\eqref{eq:interconnection2}. 
\end{proof}

Note that the Lyapunov function holds for an arbitrary number of interconnected buses and lines, as long as the \gls{dgu} controllers fulfill Prop.~\ref*{prop:bus}. 

\subsection{Model reduction and discrete time model for MPC}

The \gls{mg} model in~\eqref{eq:microgrid_dynamics} contains very fast dynamics, especially the power lines in~\eqref{eq:line_dynamics_1}\footnote{The power line dynamics constitute a first order lag with a time constant of $T = \frac{\Llj}{\Rlj}$, which takes values around $10^{-6}$ for typical medium voltage power lines \cite{tucci2016}.}. This requires very small step sizes when describing these dynamics as a discrete time system, which may compromise the real time optimization of the MPC. The solution adopted here is to neglect the inductances in the lines, i.e. to use the quasi-stationary line approximation \cite{noroozi2018model,qsl1995}. With this approximation, the \gls{mg} dynamics can be described with
\begin{subequations} \label{eq:microgrid_dynamics_withoutLines}
    \begin{align}
        \Cf \dvall & = \ifall - Y(t) \vall  - \B \Rl^{-1} \B^\top \vall \\   
        \Lf \difall & = \alpha \vall + \beta \ifall + \gamma \eall \\
        \deall & = v_{\text{ref}} - \vall . \label{eq:controller_state}
    \end{align}
\end{subequations}
These equations capture the dynamics of the passivity-based controlled \glspl{dgu} interconnected with lossy static lines. Note that~\eqref{eq:lyapunov_func} is also a Lyapunov function for~\eqref{eq:microgrid_dynamics_withoutLines}, since Prop.~\ref{th:lyapunov_timeCont} still applies due to the passivity properties of static lines (see Remark~\ref{rem:static_lines}). System~\eqref{eq:microgrid_dynamics_withoutLines} is rewritten compactly in state-space representation with 
\begin{align*}\small 
    A &= \mat{-\Cf^{-1}(Y+M\Rl M^\top) & \Cf^{-1} & 0 \\ \Lf^{-1}\alpha & \Lf^{-1}\beta & \Lf^{-1}\gamma \\ - I & 0 & 0} \in \R^{3\numOfnodes\times 3\numOfnodes} \\
    B &= \mat{0 & 0 & I}^\top \in \R^{3\numOfnodes\times \numOfnodes}
\end{align*}
and discretized using the forward Euler method with the step size $h$ 
\begin{align}\label{eq:sys_discrete}
    \xat{}{k+1} = A_k \xat{}{k} + B_k \uat{}{k}
\end{align}
with $\xat{}{k} = \col{\vall(k), \ifall(k), \eall(k)}$, $\uat{}{} = \vref(k)$, $A_k = I + hA$ and $B_k = hB $. Furthermore, the assumption of piece-wise constant system parameters (i.e. loads) during the discretization time is made. The following proposition ensures the existence of a Lyapunov function for the discrete time system~\eqref{eq:sys_discrete}, which way paves the way to apply the classical stability theorems for predictive controllers~\cite{mayne2000constrained}.
\begin{proposition} \label{th:lyapunov_timeDisc}
    Let $\lambda_i=a_i + jb_i \in \spec{A}$. There exists a $h_{\rm max}=\min_i -\frac{2a_i}{a_i^2 + b_i^2} \in \R_{>0}$ for which the system~\eqref{eq:sys_discrete} with any step size $h<h_{\rm max}$ is asymptotically stable. Moreover, a Lyapunov function $V_d(\xat{}{k}) = \xat{}{k}^\top P \xat{}{k}$ for system~\eqref{eq:sys_discrete} exists and can be computed by solving the semidefinite programming problem 
    \begin{align}
        P > 0, \qquad A_k^T P A_k + P < 0.
    \end{align}
\end{proposition} 

\begin{proof}
    If $\lambda_i=a_i + jb_i \in \spec{A}$, then $h\lambda_i+1 \in \spec{I+hA}$. From Prop.~\ref{th:lyapunov_timeCont}, we know that $a_i<0$. From $\vert h\lambda_i+1 \vert < 1 $ for all $i$ follows $h_{\rm max} = \min_i -\frac{2a_i}{a_i^2 + b_i^2}$. 
\end{proof}

%


\section{Optimal Voltage References} \label{sec:mpc}

In this section, we present the controller design used for finding suitable voltage references such that (i) the power line losses are minimized, (ii) constraint satisfaction is guaranteed, and (iii) stability is achieved. We propose first to use the classical, setpoint-tracking \gls{mpc} \cite{mayne2000constrained}, which requires the prior computation of optimal setpoints. By using the well studied tracking \gls{mpc} theory, we provide a formal stability proof using the results of Prop.~\ref{th:lyapunov_timeDisc}. Secondly, we design an equivalent economic \gls{mpc}, which is shown to be superior in performance than the classical tracking \gls{mpc}. For the economic \gls{mpc} we limit ourselves to the study the closed-loop stability by simulations in Section~\ref{sec:simulation}.

\subsection{Tracking MPC} \label{sec:mpc_track}

In this method, we compute the optimal steady-state node voltages $\vopt(k)$ in advance which is given as a setpoint to the tracking MPC. The tracking MPC then computes a $\uat{}{k}$ such that the node voltages $\vall$ follow the setpoint $\vopt$. The optimal steady-state voltages are computed solving the (open loop) optimization problem 
\begin{equation} \label{eq:ssopt}
    \ssopt(k) :=
    \left\{
    \begin{array}{cl}
         \min_{\vall} & \vall^\top \Qloss \vall \\
        \text{s. t.} & p + \Yad \, \vall \circ \vall  = 0, \\
        & \vall \in \Vconst \\
        & p \in \Pconst
    \end{array}
    \right.
\end{equation}
in which $\vall^\top \Qloss \vall = \vall^\top \B \Rl^{-1} \B^\top \vall \geq 0$ represent the line losses in the system and $p = \col{p_i} \in \R^{\numOfnodes}$ is the power infeed of a \gls{dgu} to the bus. The constraints in~\eqref{eq:ssopt} represent the power flow equations, described through the nodal admittance matrix $\Yad$ \cite{machowski2020power}, and the voltage $\Vconst \subset \R^{\numOfnodes}$ and power constraints $\Pconst \subset \R^{\numOflines}$, which may be different for each bus and \gls{dgu}. Note that the matrix $\Yad$ contains, apart from the line admittances, also the time-varying loads in~\eqref{eq:load}. With an estimation of the loads in the next time steps, this optimization problem is solved for some time steps in advance and $\vopt(k) \in \arg \min \ssopt(k)$ is obtained. Such an optimal setpoint computation is typically performed every few minutes, and a controller ensures that the steady-state is held using feedback control.

For the feedback controller, we propose to use a setpoint-tracking MPC. Having $\vopt(k)$, tracking MPC is defined as
\begin{equation} \label{eq:MPC_tracking}
    \mpctrack :=
    \left\{ 
    \begin{array}{cl}
        \min & \displaystyle \sum_{k=\kinit}^{\kend-1} \bigg\{ \Delta v(k)^\top Q \Delta v(k) + \\ 
        & \quad \Delta u(k)^\top R \Delta u(k) \bigg\} +  V_{\kend}(\Delta v(\kend)) \\
        \text{s. t.} & \xat{}{k+1} = A_k(k) \xat{}{k} + b_k \uat{}{k}, \\
        & x(k) \in \Xconst \\
        & u(k) \in \Uconst \\
        & k \in \left\{ 1,\dots, \kend \right\}, 
    \end{array}
    \right.
\end{equation}
with $\Delta v(k) = v(k) - \vopt(k)$, $\Delta u(k) = u(k) - \vopt(k)\footnote{Note that the input leading to the equilibrium point $\vopt$ is $\uat{}{}=\vopt$.}$, $R = \eta I_{\numOfnodes} \in \R^{\numOfnodes \times \numOfnodes}$, $\eta \in \R_{>0}$ and $Q = I_{\numOfnodes} \in \R^{\numOfnodes\times \numOfnodes}$, is employed to compute the optimal voltage references $\uat{}{} = \vref$ which leads to minimal losses by considering the \gls{dgu}, filter, the lines and loads. The constraints for the state and input variables in~\eqref{eq:MPC_tracking} represent the voltage and current constraints, whereas the controller state~\eqref{eq:controller_state} is not constrained. The matrices $Q$ and $R$ penalize the deviation of the voltage $\vall$ and reference voltage $\vref$ to the optimal steady-state $\vopt$. The following theorem states how to choose the terminal costs $V_{\kend}(\Delta v(\kend))$ such that closed-loop stability is ensured.
\begin{theorem}
    The closed-loop system composed of \eqref{eq:sys_discrete} with the feedback law $\uat{}{k} = u^*(1)$, ${u^*(\tau) = \arg \min \mpctrack}$ with $\tau = \{1,\dots,\kend\}$ is asymptotically stable if the terminal costs are chosen $V_{\kend}(\Delta v(\kend)) = \xat{}{\kend}^\top P \xat{}{\kend}$ with $P$ as in Prop.~\ref{th:lyapunov_timeDisc}.                                            
\end{theorem}
\begin{proof}
    Since the terminal cost $V_N(\xat{}{k}) = \xat{}{k}^\top P \xat{}{k}$ is a Lyapunov function for all $\xat{}{k} \in \R^{3\numOfnodes}$, asymptotic stability without the need of terminal constraints follows directly from the classic setpoint-tracking MPC theory \cite[Sec.~3.7.2.2]{mayne2000constrained}.
\end{proof}
\begin{remark}
    The optimal voltage setpoint for the tracking MPC is considered to be piece-wise constant for the tracking MPC. This is a classic assumption in MPC and is valid if the optimal setpoint computation is sufficiently slow compared to the system dynamics. Since the optimal setpoint computation is typically performed every few minutes and the step size of the \gls{mg} dynamics in Prop.~\ref{th:lyapunov_timeDisc} takes values of few milliseconds, this assumption is admissible and the setpoint-tracking MPC theory applies. 
\end{remark}

\subsection{Economic MPC} \label{sec:mpc_econ}

In this approach, we solve the optimization problem
\begin{equation} \label{eq:MPCecon}
    \empc := 
    \left\{
    \begin{array}{cl}
        \min_{\uat{}{}} & \displaystyle \sum_{k=\kinit}^\kend \vall(k)^\top \Qloss \vall(k) \\
        \text{s. t.} & \xat{}{k+1} = A_k(k) \xat{}{k} + b_k \uat{}{k}, \\
        & x(k) \in \Xconst \\
        & u(k) \in \Uconst \\
        & k \in \left\{ 1,\dots, \kend \right\}, 
    \end{array}
    \right.
\end{equation} 
with $\Qloss$ as in \eqref{eq:ssopt} at every time step and apply only the first control action $u^*(1)$. Hereby, the traditional control structure of computing the optimal setpoint and controlling it by feedback setpoint-tracking MPC is combined into a single feedback structure, see Fig.~\ref{fig:econVStrack}. Thus, it is not necessary to compute or know the optimal setpoint in advance, it results from the control action. It can therefore react immediately to load changes without having to compute new steady-state optimal setpoints, which is an advantage when dealing with volatile \gls{res}. However, since the objective function is not convex w.r.t. the optimal setpoint to be stabilized\footnote{Here, we are not penalizing the deviations to a setpoint, and the cost is not necessarily decreasing until the setpoint is reached \cite{muller2017economic}.}, the mature theory about classical setpoint-tracking \gls{mpc} does not hold \cite{muller2017economic}. The approach falls under the class of economic \gls{mpc}. 
\begin{figure}[t]
    \centering
    \includegraphics[]{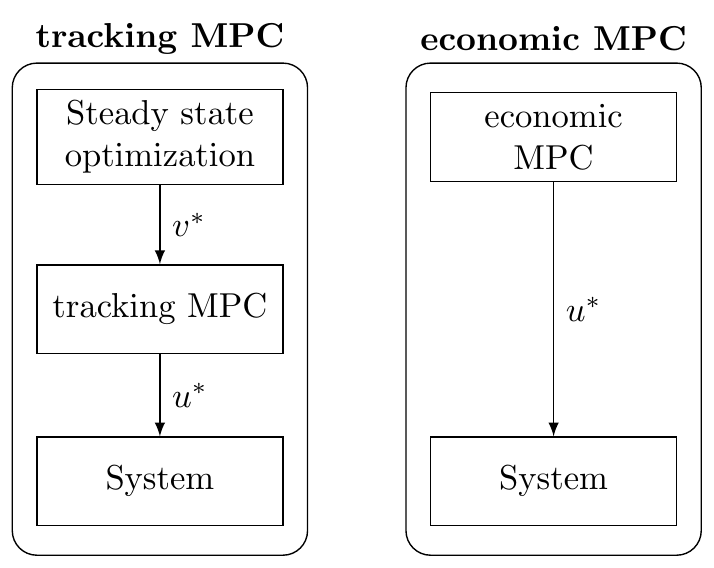}
    \caption{Comparison of setpoint tracking and economic MPC from a procedural perspective}
    \label{fig:econVStrack}
\end{figure}

In this work, the stability of economic MPC is demonstrated through simulations in the following section\footnote{Note that economic MPC is an active field of research, with almost no theoretical results when applied to power systems \cite{muller2017economic,jia2020optimal}.}. In particular, we show that economic MPC outperforms the classical control structure for predictive controllers composed of an optimal setpoint computation and a stabilizing setpoint-tracking MPC.

\section{Simulation Results} \label{sec:simulation}

In this section, simulation results for the closed-loop system with the tracking \gls{mpc}~\eqref{eq:MPC_tracking} and the economic \gls{mpc}~\eqref{eq:MPCecon} are presented. We show that the economic MPC achieves slightly better performance than the tracking MPC in the nominal case, whereas a significant performance increase is observed when disturbances unforeseen by the optimal setpoint computation in~\eqref{eq:ssopt} occur. 

The \gls{mg} considered in this work is based on the CIGRE medium voltage benchmark system, since it represents the network topology of a typical distribution system, and it is aimed to serve as a benchmark system for voltage control studies~\cite{rudion2006design}. It is composed of 11 nodes and 12 power lines and shows a meshed structure (see Fig.~\ref{fig:cigre}). Since the benchmark system is conceived as an AC system, typical DC system parameters are taken from \cite{tucci2016} for the lines and \gls{dgu} filters. The step size is chosen to be $h=\SI{10}{\milli \second}$ for both MPCs, for which Prop.~\ref{th:lyapunov_timeDisc} is fulfilled. The optimization horizon is set to $N = 300$, which corresponds to $\SI{3}{\second}$. For the tracking MPC, new optimal setpoints are computed every 30 seconds. The parameter $\eta$ is set to $\eta = 10^{-2}$ in order to achieve a better voltage tracking.

In the following, we compare both MPC approaches in the nominal case, i.e. when the load is known with no error, and in the case when load disturbances occur. Especially, we compute the transmission losses achieved in each scenario in order to assess the closed-loop performance in Section~\ref{sec:performance}.

\subsection{Scenario 1: Nominal Case}
First, the nominal case is considered, where the load is assumed to be known with no error by the predictive controllers and the steady-state optimization~\eqref{eq:ssopt} for the tracking MPC. The load during the simulation time is shown in Fig.~\ref{fig:load_nominal}. At time $t = \SI{30}{\second}$, a load step occurs in all nodes. The load step time is chosen to be at the same time when new optimal setpoints are computed for the tracking MPC, such that the MPC always receives the optimal setpoints. Note that in real applications, the optimal setpoints will likely not be computed at the same time as the load changes. The scenario presented here is hence the best case scenario for the tracking MPC in order to allow a fair comparison with the economic MPC.
\begin{figure}[ht]
    \centering
    \includegraphics{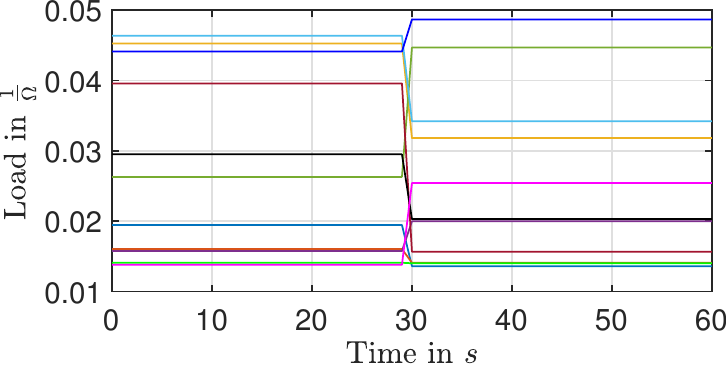} 
    \caption{Load at each bus over the simulation time in the nominal case.}
    \label{fig:load_nominal}
\end{figure}
The node voltages in the \gls{mg} with the tracking \gls{mpc} and the economic \gls{mpc} are shown in Fig.~\ref{fig:voltages_nominal}. The dot-dashed lines are the optimal setpoints computed by the steady-state optimization~\eqref{eq:ssopt}. Due to voltage, current or power flow constraints in~\eqref{eq:ssopt}, different voltage references that induce power flows through the lines are necessary. In the case of tracking \gls{mpc}, the node voltages follow the optimal setpoints (dot-dashed) accurately. Since the optimal setpoints of the voltages are computed with~\eqref{eq:ssopt} and lead to minimal power line losses, the closed-loop behavior with the tracking \gls{mpc} is (only) steady-state optimal. Only during transients, i.e. when load steps occur, considering the system dynamics for the explicit task of minimizing losses instead of regulating a setpoint may improve performance.  
\begin{figure}[ht!]
    \centering
    \includegraphics{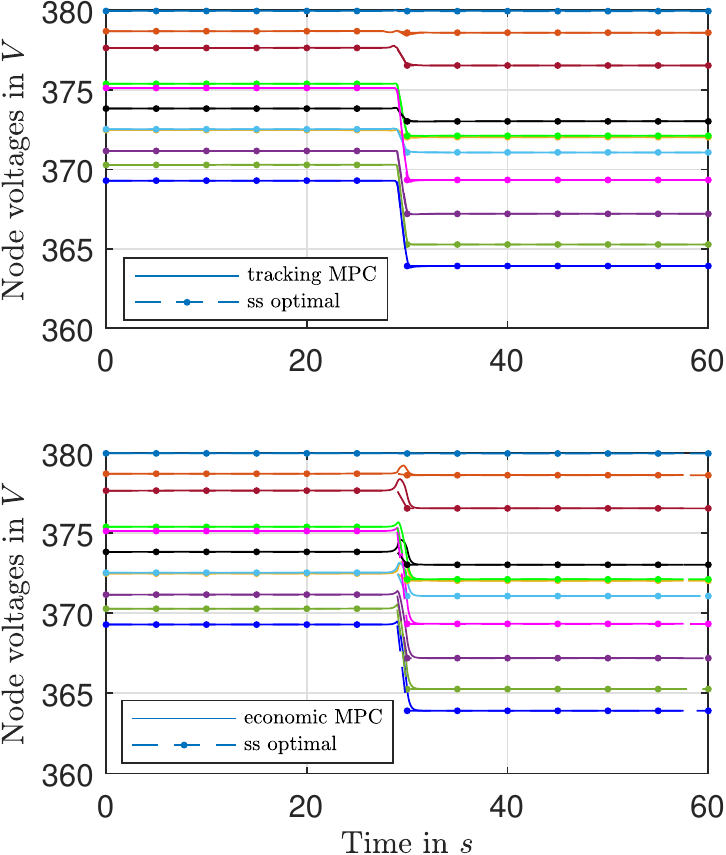} 
    \caption{Microgrid bus voltages when using a tracking MPC (above) and an economic MPC (below) in the nominal case.}
    \label{fig:voltages_nominal}
\end{figure}
The economic \gls{mpc} achieves the optimal steady-state voltages (dot-dashed) without requiring a priori voltage setpoints. The node voltages with the economic \gls{mpc} are identical to the voltages with the tracking \gls{mpc}, except for the time around the load step.

The line losses resulting from the node voltages are computed for comparing the performance of tracking and economic \gls{mpc}. The economic \gls{mpc} achieves \SI{0.2}{\percent} less transmission losses than the tracking \gls{mpc} in the nominal case (viz.~\autoref{tab:results}). This improvement is due to a better transient behavior during the load step at $t = \SI{30}{s}$. In the case of the tracking MPC, since the deviations are penalized, the node voltages follow the steps arising from the setpoint $v^*(k)$ which changes once at $t = \SI{30}{\second}$. This is, however, not optimal w.r.t. minimizing the losses. On the other hand, the economic MPC does not minimize the deviations to some setpoint, it chooses the input on the basis of minimizing the losses, also during transients. Thus, it achieves a better performance. This effect is assumed to gain importance when disturbances occur. Thus, in the next subsection, the performance of both receding horizon control approaches will be compared for the case when disturbances in the load occur.

\subsection{Scenario 2: Disturbances}
In the following, load variations are considered in order to highlight the performance of economic \gls{mpc}. We consider three types of disturbances, 1) unknown load steps, 2) load noise and 3) a line failure. The unknown load steps and load noise are depicted in Figs.~\ref{fig:load_step}~and~\ref{fig:load_noise}, respectively. The unknown load steps occur in Nodes 2 and 6, while the load noise occurs only in Node 1. Since the predictive controllers need a prediction of the load over the optimization horizon, it is assumed that the actual load is measured and considered to be constant over the optimization horizon of \SI{3}{\second}. The optimal setpoint computation in~\eqref{eq:ssopt} is assumed to not have knowledge about these measurable disturbances, since it computed the optimal setpoints in advance. The voltage trajectories are reported in this section, while the performance comparison is made in Section~\ref{sec:performance}. 
\subsubsection{Unknown load steps}
The node voltages in the \gls{mg} equipped with both predictive controllers can be seen in Fig.~\ref{fig:voltages_step}. The tracking MPC still achieves an acceptable regulating behavior when the load step occurs, since it penalizes the deviations to that given setpoint. However, note that the optimal setpoint computed in advance is no longer optimal due to the load steps, which are not known in advance and not considered in~\eqref{eq:ssopt}. On the other hand, the node voltages set by the economic MPC differ w.r.t. the setpoints from~\eqref{eq:ssopt} when the disturbances occur. 
Since the economic MPC does not minimize the deviations to a given setpoint and instead directly minimizes the losses, a new but optimal steady-state arises. 
%
\begin{figure}[ht!]
    \centering 
    \includegraphics{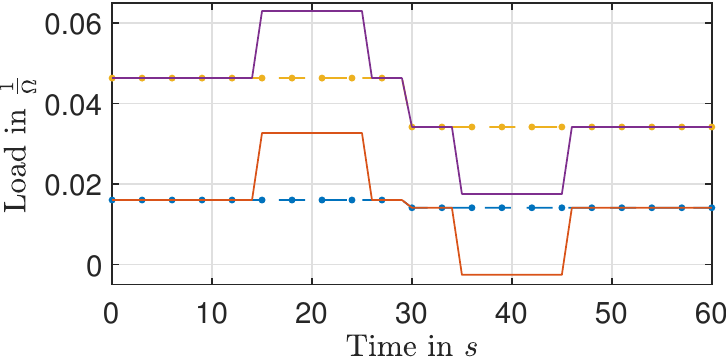} 
    \caption{Predicted (dot-dashed) and real load in the case of unknown steps in Node 2 and 6.} 
    \label{fig:load_step}
\end{figure}
\begin{figure}[ht!]
    \centering 
    \includegraphics{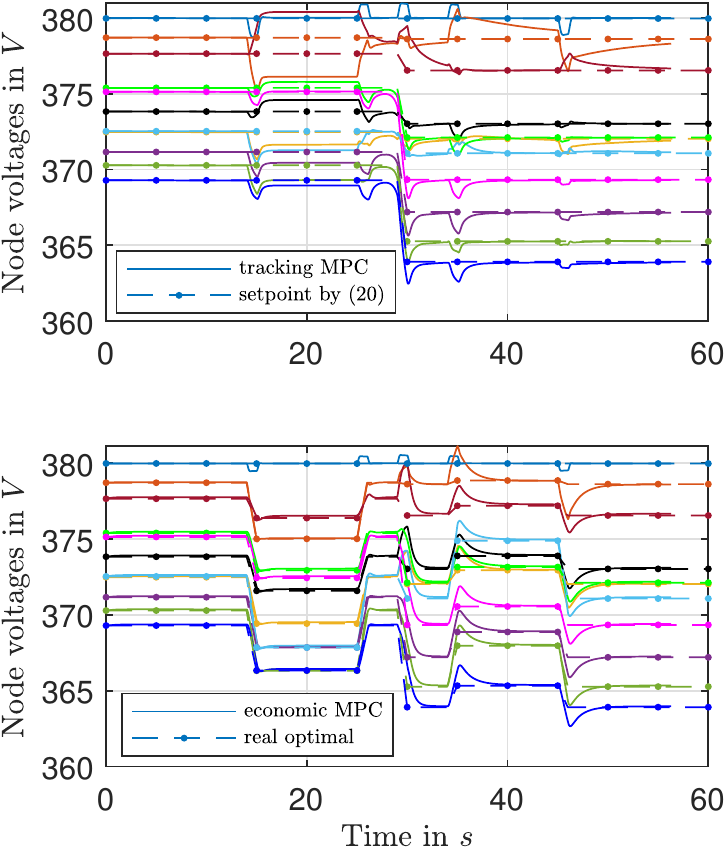} 
    \caption{Microgrid bus voltages when using a tracking MPC (above) and an economic MPC (below) under an unknown load step.}
    \label{fig:voltages_step}
\end{figure}
\subsubsection{Load noise}
The node voltages in the \gls{mg} with both receding horizon controllers in the case of load noise are shown in Fig.~\ref{fig:voltages_noise}. The voltages set by the tracking MPC show significant oscillations around the given optimal setpoints. In contrast, the node voltages produced by the economic MPC are smooth and slightly differ from the optimal steady-state voltages computed in advance, as expected. 
\begin{figure}[ht!]
    \centering 
    \includegraphics{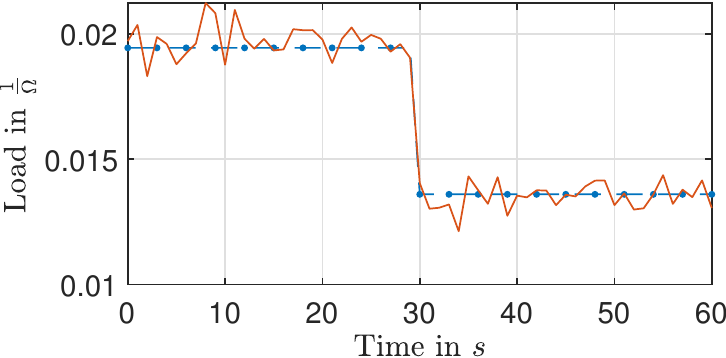} 
    \caption{Predicted (dot-dashed) and real load in the case of load noise in Node 1.}
    \label{fig:load_noise}
\end{figure}
\begin{figure}[ht!]
    \centering 
    \includegraphics{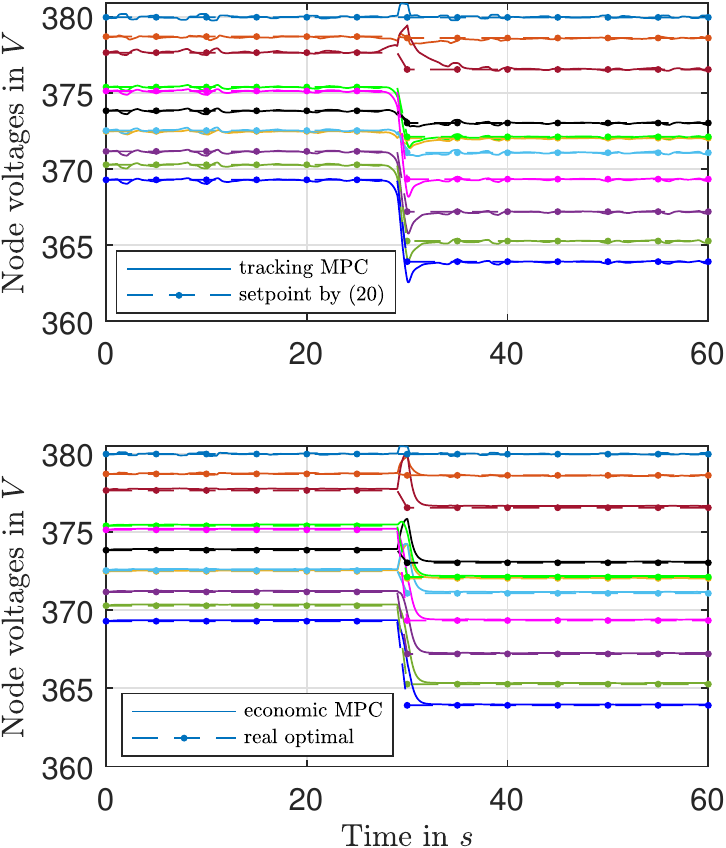} 
    \caption{Microgrid bus voltages when using a tracking MPC (above) and an economic MPC (below) under unknown load noise.}
    \label{fig:voltages_noise}
\end{figure}

\subsubsection{Line failure}
In the scenario of a line failure, the line between Node 3 and 8 fails at time $t_{\rm fail} = \SI{20}{\second}$ and is unavailable thereafter. The line failure is assumed to be measurable and known by the predictive controllers, but not by the optimal setpoint computation in~\eqref{eq:ssopt} which happens in advance. The node voltages in the \gls{mg} with both predictive controllers in the case of a line fail can be seen in Fig.~\ref{fig:voltages_lineFail}. The tracking MPC minimizes the deviation to the provided setpoint (dot-dashed line), which is not optimal for the new \gls{mg} topology. The voltages when using the economic MPC converge to a new steady-state, which is optimal under the new \gls{mg} configuration.

\begin{figure}[ht!]
    \centering 
    \includegraphics{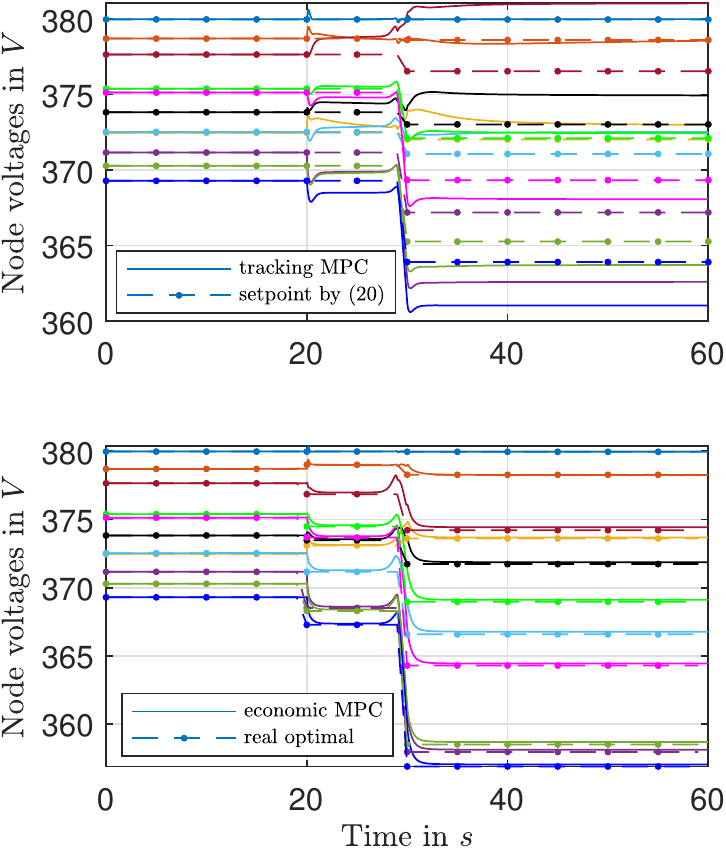} 
    \caption{Microgrid bus voltages when using a tracking MPC (above) and an economic MPC (below) under a line failure.}
    \label{fig:voltages_lineFail}
\end{figure} 

\subsection{Performance comparison} \label{sec:performance}
The performance increase of the economic MPC w.r.t. the tracking MPC is shown in Table~\ref{tab:results}. It shows the reduction of transmission losses for the described scenarios normalized to the transmission losses from the economic MPC. The economic MPC achieves a better performance in all cases. The increase in performance depends clearly on the severity of the disturbance. The small disturbance from the load noise shows a 2.6 \% less transmission losses. For greater disturbances, the reduction increases to 9.3 \% (load step), and 15.6 \% (line failure). Thus, greater disturbances yield greater performance benefits from the economic MPC, since the optimal voltages differ more from the a priori calculated setpoints. Since the economic MPC does not need to compute any optimal setpoint in advance, it has a clear advantage also in terms of computation and enjoys a simpler feedback control structure naturally achieving the optimal setpoints as a result of the control action.
\begin{table}[!ht]
    \centering
    \caption{Performance increase of economic MPC over tracking MPC in different scenarios}
    \begin{tabular}{c c c c} 
        \toprule
        Scenario & Performance increase of economic MPC \\
        \midrule
        Nominal & 0.2 \% \\
        Unknown load steps &  9.3 \% \\
        Unknown load noise & 2.6 \%  \\
        Line failure & 15.6 \%  \\
        \bottomrule
    \end{tabular}
    	 
    \label{tab:results}	
\end{table}

\section{Conclusion}

This paper presents two MPC-based controllers for computing optimal voltage references for a DC \gls{mg} with passivity-based primary controllers. An asymptotically stabilizing setpoint-tracking MPC is designed, and stability is proven by considering the passivity properties of the DC \gls{mg} with the primary controller. Furthermore, we present an economic MPC controller which outperforms the traditional control scheme composed of optimal steady-state computation and setpoint-tracking MPC. 

Future work may concern the rigorous stability proofs of the economic MPC for DC \glspl{mg}. 
Dropping the necessity of computing a priori optimal setpoints is a key advantage of economic MPC which can be exploited for distributed control of large scale energy systems.

\printbibliography

\end{document}